\documentclass[12pt,a4paper]{amsart}

\usepackage{euscript,amsfonts,amssymb,amsmath,amscd,amsthm,enumerate, hyperref}

\usepackage{comment}

\usepackage{MnSymbol}

\usepackage{skak}

\usepackage{graphicx}

\usepackage{color}

\usepackage[margin=1.1in]{geometry}
\newtheorem{theorem}{Theorem} 
\newtheorem*{theorem*}{Theorem}
\newtheorem{lemma}[theorem]{Lemma}

\theoremstyle{remark}
\newtheorem{remark}[theorem]{Remark}

\newcommand{\la}{\lambda}

\newcommand{\eps}{\varepsilon}

\newcommand{\all}{\mathbf{+}}
\newcommand{\nul}{\zugzwang}
\newcommand{\rr}{\mathbf{--}}
\newcommand{\uu}{\mathbf{|}}
\newcommand{\ur}{\mathbf{ \lefthalfcap }}
\newcommand{\ru}{\mathbf{ \righthalfcup }}

\newcommand{\s}{\mathcal S}

\DeclareMathOperator{\typ}{type}
\DeclareMathOperator{\Th}{Th}
\DeclareMathOperator{\Si}{Sign}

\title[Dynamics on the 6-vertex model]
{An irreversible local Markov chain that preserves the six vertex model on a torus}

\author{Alexei Borodin}

\address[Alexei Borodin]{ Massachusetts Institute of Technology, Cambridge, USA, and Institute for
Information Transmission Problems of Russian Academy of Sciences, Moscow, Russia. E-mail: borodin@math.mit.edu}

\author{Alexey Bufetov}

\address[Alexey Bufetov]{ Massachusetts Institute of Technology, Cambridge, USA, and International Laboratory of Representation Theory and Mathematical Physics, National Research University Higher School of Economics, Moscow, Russia. E-mail: alexey.bufetov@gmail.com}

\begin{document}

\begin{abstract}
We construct an irreversible local Markov dynamics on configurations of up-right paths on a discrete two-dimensional torus, that preserves the Gibbs measures for the six vertex model. An additional feature of the dynamics is a conjecturally nontrivial drift of the height function.

\end{abstract}

\maketitle

\section{Introduction} Random growth models is a rapidly developing subject that focuses on studying large-time behavior of randomly growing interfaces in $(d+1)$ dimensions. The growth mechanism is usually assumed to be local, in the sense that distant parts of the interface evolve (almost) independently, to conform with the belief that most inter-molecular interaction mechanisms in nature are local. 

Not surprisingly, the best understood case is $d=1$. Many results are available including the hydrodynamic (law of large numbers) behavior for certain classes of systems, see e.g. the book of Kipnis-Landim \cite{KipnisLandim} and reference therein, as well as fluctuations for certain \emph{integrable} systems in the KPZ universality class, see e.g. Ferrari--Spohn \cite{FerrariSpohn}, Corwin \cite{Corwin-KPZ}, and references therein. 

One key fact that is very useful in the (1+1)-dimensional situation is that the steady states of the growth model often turn out to be described by product-measures. The lack of such a simple structure in higher dimensions is one obstacle for studying the $d>1$ case. Furthermore, it is a challenge to find any irreversible (or driven) Markov evolution of an interface in $(d+1)$-dimensions with $d>1$ for which one could describe the steady states. \footnote{We omit here the case of heat bath or Glauber-type dynamics that are typically reversible, have asymptotically zero drift, and belong to different universality classes than the irreversible examples.} See Gates-Westcott \cite{GatesWestcott} for a notable exception. 

Recent developments in studies of integrable probabilistic models brought up new examples. In \cite{BorFer}, \cite{BorFer2}, Borodin-Ferrari constructed a local random growth model of stepped surfaces in $(2+1)$d for which the asymptotic growth velocity (as well as local correlations and global Gaussian fluctuations) were computed explicitly, see also Borodin-Bufetov-Olshanski \cite{BorodinBufetovOlshanski} for a generalization that includes similar results for various initial conditions. 

Very recently, Toninelli \cite{Toninelli} proved that the dynamics of Borodin-Ferrari, as well as its generalization that acts on dimers on the square lattice, can be correctly defined on translation-invariant random surfaces over the whole plane. 
He also proved that members of a particular two-dimensional family of such random surfaces are steady states for the resulting dynamics. In fact, this two-dimensional family exactly coincides with the two-dimensional family of translation-invariant Gibbs measures for dimers on either hexagonal or square lattice, cf. Kenyon-Okounkov-Sheffield \cite{Kenyon-Okounkov-Sheffield}, Sheffield \cite{Sheffield}. While checking similar statements on the torus in this case is fairly straightforward, extending to the whole plane is very nontrivial and requires substantial efforts.

The correlation functions of the measures from this two-dimensional family are determinantal. These measures can be thought of as originating from free-fermion (dimer) models. In the last few years certain progress has been achieved in understanding random growth models in (1+1)d that are not free fermionic (but integrable), see e.g. \cite{BorPetLectures} and references therein. It is enticing to try to expand this progress into the domain of (2+1)-dimensional models. 

One step in this direction was taken by Corwin-Toninelli \cite{CorwinToninelli}. They used a triangular array dynamics from Borodin-Corwin \cite{BorCor} (which generalized  \cite{BorFer} and \cite{Bor-schur} to a non-free fermion case) to define an irreversible Markov dynamics on particles living on a two-dimensional torus. They found that Gibbs measures with respect to a particular explicit $q$-weighting are preserved by the dynamics. 

In this work we do something similar, but the extension of \cite{BorFer}, \cite{Bor-schur} is different, and it is not a part of the Macdonald processes framework of \cite{BorCor}. 

We construct an explicit irreversible Markov dynamics on up-right path configurations on a discrete two-dimensional torus. Such path configurations can be viewed as level lines of a {height function} whose plot is the evolving interface. Our main result states that the Gibbs measures of the celebrated six vertex model are invariant with respect to our dynamics. 

The literature on the six vertex model is vast, and we will not even try to survey it. An interested reader could consult e.g. Palamarchuk-Reshetikhin \cite{PalamarchukReshetikhin} and references therein. Its Gibbs measures on a torus are parametrized by two nonnegative integers (in addition to the weights of the six vertex model itself). 

There remains an open and very interesting question to see if the dynamics defined in this paper can be extended to the whole plane while preserving the two-dimensional family of the six vertex translation invariant Gibbs measures. Many things are known about such Gibbs measures but most of them are conjectural,  although widely accepted in the physics literature. One might hope that having an explicit dynamics that preserves them could shed more light on their properties. 

In the text below we give a verification proof of our main result. A brief explanation of its derivation is given in the Appendix; it is based on properties of a family of rational symmetric functions introduced in \cite{Bor}. Curiously, while the positivity of the measures participating in the derivation requires that the parameters of the six vertex model lie in the so-called ferro-electric region, the final result knows nothing about this restriction --- the  Markov dynamics is positive and preserves the Gibbs measures for arbitrary nonnegative values of the six vertex weights of the model. 

\subsection*{Acknowledgments}

A. Borodin was partially supported by the NSF grant DMS-1056390. A. Bufetov was partially supported by "Dynasty" foundation, the RFBR grant 13-01-12449, and by the Government of the Russian Federation within the framework of the implementation of the 5-100 Programme Roadmap of the National Research University Higher School of Economics.

\section{Preliminaries}

Let $\Th$ be the square grid on a two-dimensional torus of size $M \times N$. It can be represented by the (conventional) grid $\{0,1,\dots, M\} \times \{0,1,\dots, N\}$ with the identification of vertices $(0,i) = (M,i)$ and $(i,0)=(i,N)$, $i=0,1, \dots$. We say that two vertices are adjacent if they connected by an edge in $\Th$; we say that the vertices are \textit{h-adjacent} if this edge is horizontal (the first coordinate changes) and \textit{v-adjacent} if this edge is vertical.
We will use the circular order in the horizontal and vertical directions that is inherited from the sets $\{0,1,\dots, M-1\}$ and $\{0,1,\dots, N-1\}$. All relevant operations and notions (like ``$a+1$'' or ``the top vertex from v-adjacent vertices'', etc) are taken in accordance with these circular orders.

A \textit{state} of a 6-vertex model on $\Th$ is a subset of edges from $\Th$ such that each vertex has one of 6 allowed types. We will think of a horizontal edge from this subset as an arrow directed to the right and of a vertical edge as an arrow directed to the top. The 6 allowed types of vertices and our notation for them are shown in Figure \ref{fig:6types}, and an example of a state is shown in Figure \ref{fig:state}. To each of these types of vertices we assign a real positive \textit{weight}; we denote these numbers by $w(\nul)$, $w(\all)$, $w(\ur)$, $w(\ru)$, $w(\uu)$, and $w(\rr)$. The \textit{weight} of a state is defined as the product of weights of all vertices from $\Th$.

\begin{figure}
\includegraphics[width=12cm]{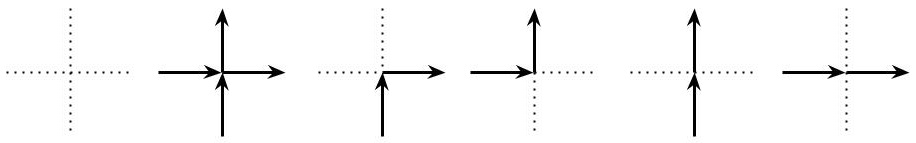}
\caption{We denote these 6 types of vertices by symbols $\nul$, $\all$, $\ur$, $\ru$, $\uu$, $\rr$, respectively. }
\label{fig:6types}
\end{figure}

\begin{figure}
\includegraphics[width=12cm]{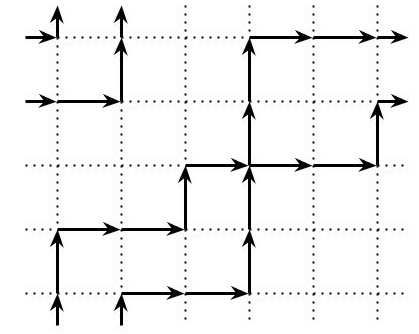}
\caption{An example of a state on $\Th$ with $M=6$ and $N=5$.}
\label{fig:state}
\end{figure}

For a state $s$, let $N_s (X)$ denote the number of vertices of the type $X$ in $s$ (here $X$ can be one of 6 allowed types of vertices). For a vertex $A$ from $\Th$ we denote by $\typ_s (A)$ the type of this vertex in the state $s$.

It is clear that the number of horizontal arrows in any state equals $M k_1$, for some $k_1 \in \mathbb N \cup \{0\}$, and the number of vertical edges equals $N k_2$, for some $k_2 \in \mathbb N \cup \{0\}$. Let $\s_{k_1,k_2}$ be the set of all states with fixed numbers $k_1, k_2$. The \textit{Gibbs measure} on $\s_{k_1,k_2}$ is defined by
$$
\mathrm{Prob} (s) := \frac{w(s)}{\sum_{s \in \s_{k_1,k_2}} w(s)}, \qquad s \in \s_{k_1,k_2}.
$$
If all vertex weights are positive, then this is a probability measure on $\s_{k_1,k_2}$. Though all 6 weights of vertices can be arbitrary positive numbers, this measure depends on only two free parameters; see the end of Section \ref{sec:description} for more details.

\section{Description of the dynamics}
\label{sec:description}

The goal of this section is to describe our dynamics $\mathcal D$. This will be a continuous time Markov chain on the finite set $\s_{k_1,k_2}$. We will consider the case $1 \le k_1 \le N-1, 1 \le k_2 \le M-1$ only. In the remaining cases $k_1=0,N$ or $k_2=0,M$ the set of states becomes rather poor, and our dynamics degenerates to a one-dimensional dynamics (ASEP) on it.

For the description it is enough to define the matrix elements $p(s_1,s_2)$, $s_1, s_2 \in \s_{k_1,k_2}$ of the generator of our Markov dynamics. Since $k_1, k_2$ are fixed, we will omit the dependence on them in the sequel.

\begin{figure}
\includegraphics[width=12cm]{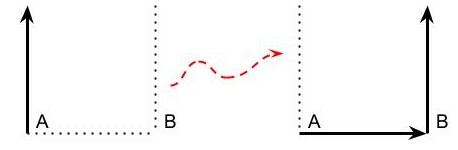}
\caption{Jumps to the right happen when there are two vertices with such a structure of arrows; arrows on edges not in the picture can be arbitrary. The four different cases are shown in Figures \ref{fig:rightJump-1} --- \ref{fig:rightJump-4}. }
\label{fig:rightJump-common}
\end{figure}

Let us proceed with a formal description. A jump to the right can occur when there is a configuration of arrows depicted in the left part of Figure \ref{fig:rightJump-common} (here and in the sequel we assume that arrows on the edges that are not present in the figures can be arbitrary). We want to single out all types that vertices $A$ and $B$ from Figure \ref{fig:rightJump-common} can have.
Let us call the pairs of h-adjacent vertices depicted in the left panels of Figures \ref{fig:rightJump-1} --- \ref{fig:rightJump-4} \textit{r-admissible pairs}.

In a similar vein, we call the pairs of h-adjacent vertices that are depicted on the left panels of Figures \ref{fig:leftJump-1} --- \ref{fig:leftJump-4} \textit{l-admissible pairs}.

Let us define the possible transform of a state $s \in \s$; there are two types of such transforms --- for jumps to the right and for jumps to the left.

Let us start with describing the possible transforms that correspond to the right jumps. Let $A$ and $B$ be an h-adjacent pair of vertices from Figures \ref{fig:rightJump-1} --- \ref{fig:rightJump-4}: Let us denote by $C$ the vertex at the same vertical line as $A$ such that all vertices between $A$ and $C$ have type $\uu$ and $C$ has type $\ur$ or $\all$. Note that such a vertex always exist since $k_1 \ge 1$ by our assumption, i.e. there should be a horisontal arrow in $s$ that starts on the vertical line which contains $A$. Let $D$ be the vertex that is h-adjacent to $C$ and is in the same vertical line as B (see the left panel of Figure \ref{fig:rightJump-column} for the pictorial explanation of these definitions).

The possible transform corresponding to the jump to the right changes the edges in the rectangle ABCD as shown in Figure \ref{fig:rightJump-column}; the state of all other edges remains the same.

Let us now describe the possible transforms that correspond to the jumps to the left. Let $A$ and $B$ be an h-adjacent pair of vertices from Figures \ref{fig:leftJump-1} --- \ref{fig:leftJump-4}: Denote by $C$ the vertex at the same vertical line as $A$
such that all vertices between $A$ and $C$ have the type $\rr$ and $C$ has the types $\nul$ or $\ru$. Note that such a vertex always exist since $k_1 \le N-1$ by our assumption, i.e. there should be a horisontal edge free of arrows of $s$ that ends on the vertical line which contains $A$. Let $D$ be the vertex that is h-adjacent to $C$ and is in the same vertical line as $B$ (see the left panel of Figure \ref{fig:leftJump-column} for the pictorial explanation of these definitions).

The possible transform corresponding to the jump to the left changes the edges in the rectangle ABCD as shown in Figure \ref{fig:leftJump-column}; the state of all other edges remains the same.

\begin{figure}
\includegraphics[width=14cm]{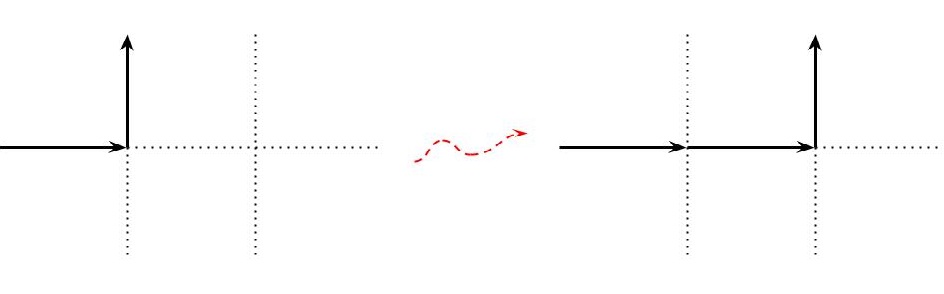}
\caption{Admissible pair; jump to the right; rate = $\sqrt{\frac{w(\rr) w(\uu)}{w(\all) w(\nul)}}$.}
\label{fig:rightJump-1}
\end{figure}

\begin{figure}
\includegraphics[width=14cm]{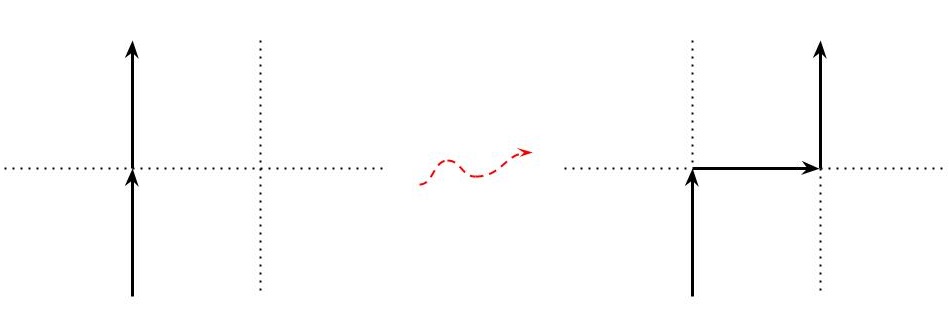}
\caption{Admissible pair; jump to the right; rate = $\frac{w(\ru) w(\ur)}{\sqrt{w(\all) w(\nul) w(\uu) w(\rr)}}$.}
\label{fig:rightJump-2}
\end{figure}

\begin{figure}
\includegraphics[width=14cm]{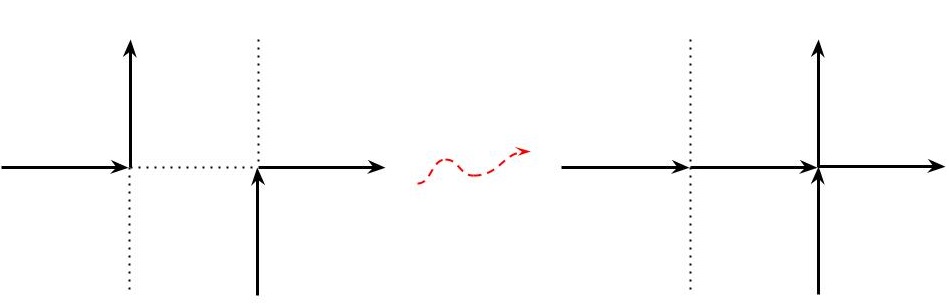}
\caption{Admissible pair; jump to the right; rate = $\frac{\sqrt{w(\all) w(\nul) w(\uu) w(\rr)}}{w(\ru) w(\ur)}$.}
\label{fig:rightJump-3}
\end{figure}

\begin{figure}
\includegraphics[width=14cm]{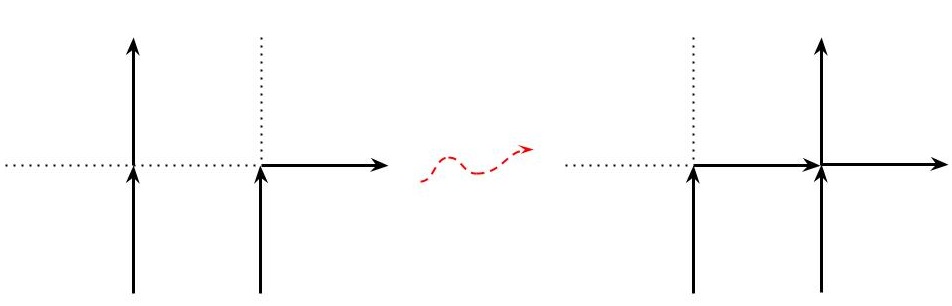}
\caption{Admissible pair; jump to the right; rate = $\sqrt{\frac{w(\all) w(\nul)}{w(\rr) w(\uu)}}$ .}
\label{fig:rightJump-4}
\end{figure}


\begin{figure}
\includegraphics[width=14cm]{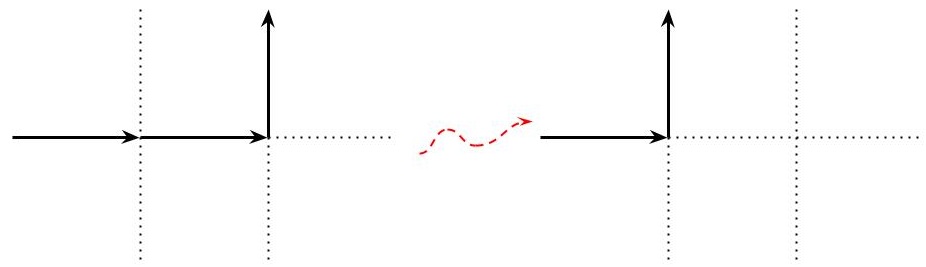}
\caption{Admissible pair; jump to the left; rate = $\sqrt{\frac{w(\all) w(\nul)}{w(\rr) w(\uu)}}$.}
\label{fig:leftJump-1}
\end{figure}

\begin{figure}
\includegraphics[width=14cm]{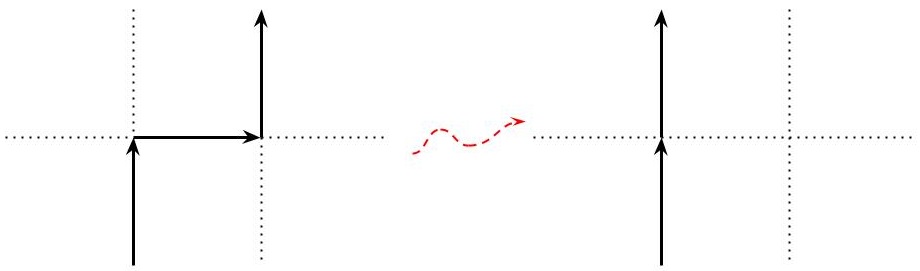}
\caption{Admissible pair; jump to the left; rate = $\frac{\sqrt{w(\all) w(\nul) w(\uu) w(\rr)}}{ w(\ru) w(\ur)}$.}
\label{fig:leftJump-2}
\end{figure}

\begin{figure}
\includegraphics[width=14cm]{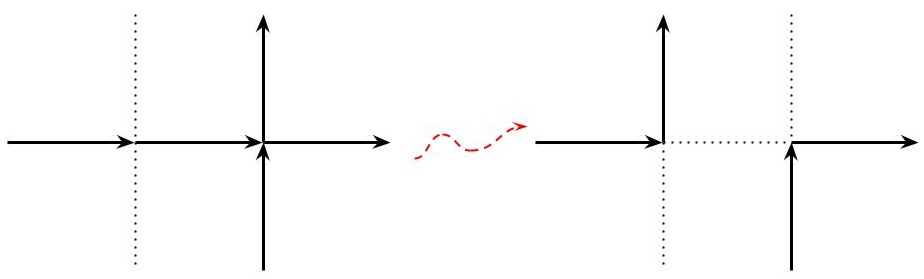}
\caption{Admissible pair; jump to the left; rate = $\frac{w(\ru) w(\ur)}{\sqrt{w(\all) w(\nul) w(\uu) w(\rr)}}$.}
\label{fig:leftJump-3}
\end{figure}

\begin{figure}
\includegraphics[width=14cm]{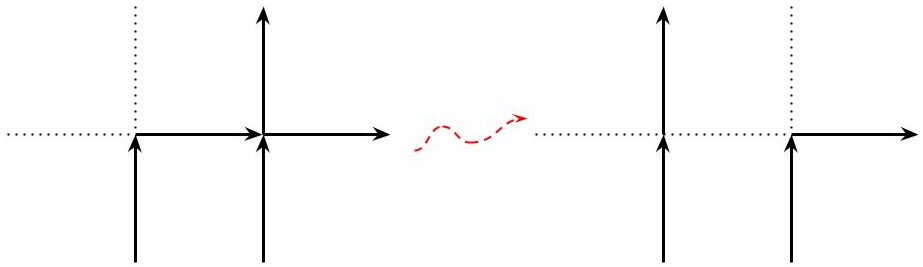}
\caption{Admissible pair; jump to the left; rate = $\sqrt{\frac{w(\rr) w(\uu)}{w(\all) w(\nul)}}$.}
\label{fig:leftJump-4}
\end{figure}

\begin{figure}
\includegraphics[width=14cm]{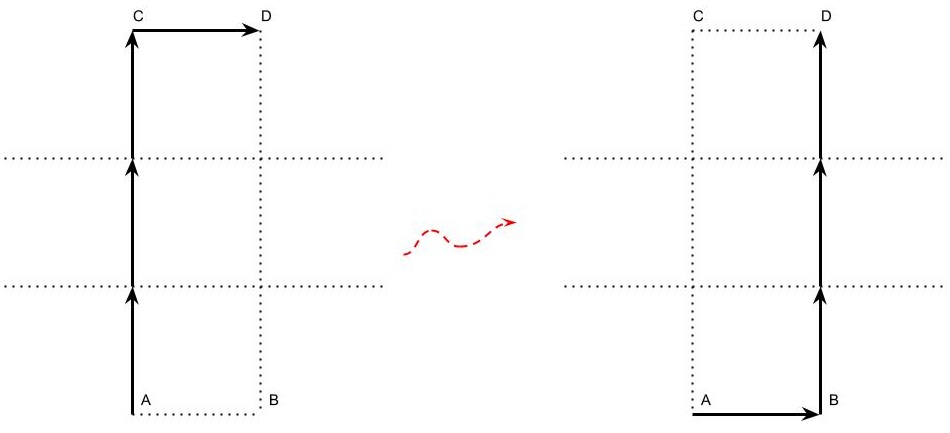}
\caption{Jump to the right results in such a change of vertices. The vertices between A and C can have type $\uu$ only; however, there can be arbitrarily many of them. Not depicted arrows can be arbitrary and remain the same.}
\label{fig:rightJump-column}
\end{figure}

\begin{figure}
\includegraphics[width=14cm]{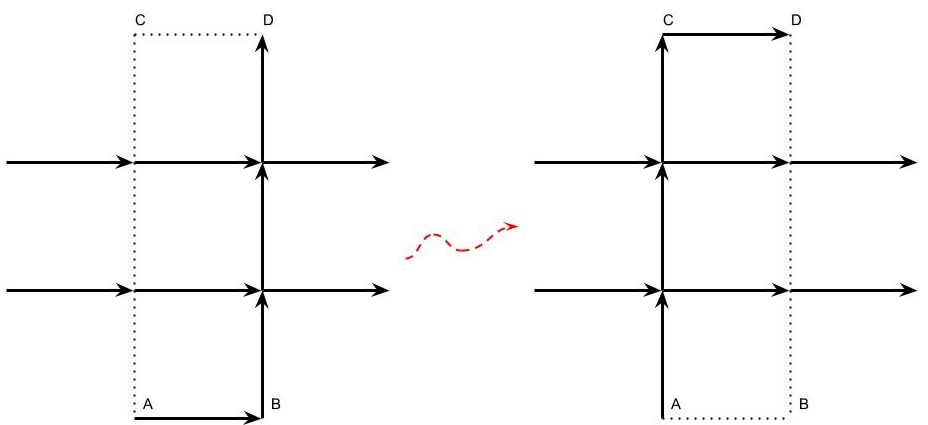}
\caption{Jump to the left results in such a change of vertices. The vertices between A and C can have type $\uu$ only; however, there can be arbitrarily many of them. Not depicted arrows can be arbitrary and remain the same.}
\label{fig:leftJump-column}
\end{figure}

Now we are in a position to define our Markov dynamics $\mathcal D$.
The transition rate of our dynamics $p(s_1 \to s_2)$, $s_1, s_2 \in \s$, is positive if and only if one can chose an admissible pair of vertices in $s_1$ such that the corresponding possible transform turns $s_1$ into $s_2$. The value of $p( s_1 \to s_2)$ depends only on the type of admissible pairs and is given in Figures \ref{fig:rightJump-1} --- \ref{fig:leftJump-4} for each of them.

Let us comment on our exact formulas for the transition rates.

First, note that the $r$-admissible and $l$-admissible pairs are in the natural bijection: The jump to the right in Figure \ref{fig:rightJump-1} ( \ref{fig:rightJump-2}, \ref{fig:rightJump-3}, \ref{fig:rightJump-4}) is exactly the reverse to the jump to the left on Figure \ref{fig:leftJump-1} (\ref{fig:leftJump-2}, \ref{fig:leftJump-3}, \ref{fig:leftJump-4}, respectively). The jump rates for corresponding pairs are also inverse to each other. Second, the four different jump rates for the jumps to the right are also split into two pairs which are inverses of each other. Note also that the values of jump rates depend only on local properties of a state. However, the jumps themselves can propagate upwards arbitrarily far.

The statements of the next two paragraphs will not be used in the sequel, but it is natural to address them here.

There is another bijection between jumps to the right and to the left. For a state $s$ one can define a \textit{dual} state which is defined as the set of all edges except for the edges from $s$. Then the jumps to the left (Figures \ref{fig:leftJump-1}, \ref{fig:leftJump-2}, \ref{fig:leftJump-3}, \ref{fig:leftJump-4}) are in correspondence with the jumps of the dual state to the right ( Figures \ref{fig:rightJump-4}, \ref{fig:rightJump-3}, \ref{fig:rightJump-2}, \ref{fig:rightJump-1}, respectively). The jump rates for corresponding pairs coincide.

The jump rates for any admissible pair can be computed with the use of dual states. The recipe is: Multiply the weights of the pair of vertices and their weights in the dual state (four factors in total) after the transform and divide by the same product computed before the transform; then take the square root. For example, for an admissible pair from Figure \ref{fig:rightJump-1} this computation looks as follows. After the transform (the right panel of Figure \ref{fig:rightJump-1}) we need to take $w(\rr) w(\ru)$ --- the weight of these vertices --- and then multiply it to $w(\uu) w(\ur)$ --- the weight of these vertices in the dual state. Analogously, before the transform (the left panel of Figure \ref{fig:rightJump-1}) we obtain $w(\ru) w(\nul) \times w(\ur) w(\all)$. Dividing these quantities and taking the square root we obtain the jump rate for this admissible pair from the caption of Figure \ref{fig:rightJump-1}.

We have defined all jump rates for our dynamics; thus, our description of the dynamics $\mathcal D$ is complete now. Let us add more comments.

We operate with 6 weights for vertices without any constraints; however, different weights can give the same Gibbs measure. We have 4 equalities

1) $N_s (\all) + N_s (\ru) + N_s (\rr) = const$ --- because all configurations must have the same number of horizontal arrows.

2) $N_s (\all) + N_s (\ru) + N_s (\uu) = const$ --- because all configurations must have the same number of vertical arrows.

3) $N_s (\nul) + N_s (\all) + N_s (\ru) + N_s (\ur) + N_s (\rr) + N_s (\uu) = const$ --- because the number of all vertices is the same for all configurations.

4) $N_s (\ru) - N_s (\ur) = 0$ --- see the proof of Lemma \ref{lem:flip-weight} below.

This gives rise to 4 possible types of transforms of weights that produce the same Gibbs measure. They are as follows:

1) Multiply $w(\all)$, $w(\ru)$, and $w(\rr)$ by $C_1$.

2) Multiply $w(\all)$, $w(\ru)$, and $w(\uu)$ by $C_2$.

3) Multiply the weights of all vertices by $C_3$.

4) Multiply $w(\ru)$ by $C_4$ and multiply $w(\ur)$ by $C_4^{-1}$.

One easily checks that our jump rates are invariant under these transforms; it means that for the same Gibbs measure we have the same dynamics.

We also note that since we have 6 weights for the types of vertices and 4 possible transforms of them, then in fact we have only 2 free parameters of the model coming from weights; we also have two integer parameters $k_1$ and $k_2$. We have no further restrictions: If the weights of vertices are positive, then all our jump rates are positive as well.

\section{Invariance of the Gibbs measure}

In this section we establish the main theorem of the present paper, Theorem \ref{th:main} below.
We start with some necessary notions and lemmas.

For any state $s \in \s$ we define a state $\bar s$ which we call the \textit{flip} of $s$. By definition, this is the state which is obtained from $s$ by changing all arrows into the opposite direction and by using the opposite circular order in both horizontal and vertical directions. In other words, we rotate the grid into 180 degrees and change the direction of all arrows. See an example in Figure \ref{fig:flip}.

\begin{figure}
\begin{center}
\noindent{\scalebox{0.6}{\includegraphics{dyn6vert-state.jpg}}} \hskip 1cm {\scalebox{0.6}{\includegraphics{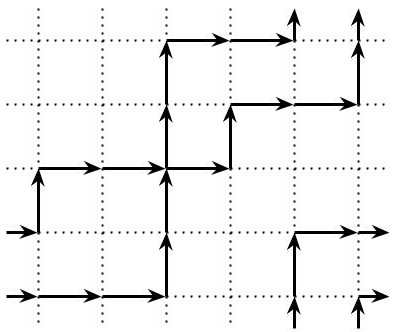}}}
\caption{An example of a state and its flip on a torus grid.}
\label{fig:flip}
\end{center}
\end{figure}

\begin{lemma}
\label{lem:flip-weight}
For any $s \in \s$ we have $w(s) = w( \bar s)$.
\end{lemma}
\begin{proof}
Note that vertices of types $\all$, $\nul$, $\uu$, and $\rr$ are invariant under the flip, while $\ur$ turns into $\ru$ and vice versa. Therefore, it is enough to prove that $N_s (\ur) = N_s (\ru)$ for any $s \in \s$. But this is clearly visible from the up-right paths structure of $s$: Along each up-right path the number of right turns equals the number of left turns, and each vertex $\all$ contains one left turn and one right turn inside.
\end{proof}

\begin{lemma}
\label{lem:flip-main}
Let $s_1, s_2 \in \s$. If $p(s_1 \to s_2) >0$, then $p(\bar s_2 \to \bar s_1) >0$ as well. Moreover, we have
\begin{equation}
\label{eq:flip-main-prop}
w(s_1) p(s_1 \to s_2) = w(\bar s_2) p (\bar s_2 \to \bar s_1).
\end{equation}
\end{lemma}
\begin{proof}
Let us consider the case when $s_1$ transforms into $s_2$ by a right jump.

From the construction, the two top vertices --- C and D --- in which the transformation happens can look in 4 possible ways (see Figure \ref{fig:rightJump-column} and Figure \ref{fig:endVert}). It is readily visible that after the flip they exactly coincide with admissible pairs of the jumps to the right from $\bar s_2$ to $\bar s_1$ (left panels of Figures \ref{fig:rightJump-1}, \ref{fig:rightJump-2}, \ref{fig:rightJump-3}, \ref{fig:rightJump-4}, respectively). This proves the first statement of the lemma.

\begin{figure}
\includegraphics[width=15cm]{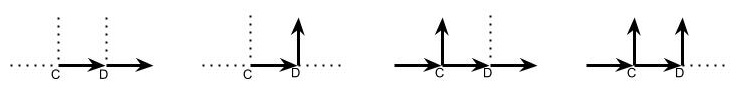}
\caption{The 4 possible ways how the non-depicted arrows in C and D from Figure \ref{fig:rightJump-column} can look like.}
\label{fig:endVert}
\end{figure}

Let us prove now the second claim. Lemma \ref{lem:flip-weight} asserts that $w(\bar s_2) = w(s_2)$. Note that the set of types of vertices between $A$ and $C$, and between $B$ and $D$ remains the same during the transform (see Figure \ref{fig:rightJump-column}). Thus, we have
$$
w(s_2) = w(s_1) \frac{w (\typ_{s_2} (A) ) w(\typ_{s_2} (B) ) w (\typ_{s_2} (C) ) w(\typ_{s_2} (D) )}{w(\typ_{s_1} (A)) w(\typ_{s_1} (B)) w(\typ_{s_1} (C)) w(\typ_{s_1} (D))},
$$
since we need to control the change of types in A,B,C,D only.
Substituting these equalities into \eqref{eq:flip-main-prop} we see that it remains to prove
$$
\frac{p(s_1 \to s_2)}{p(\bar s_2 \to \bar s_1)} = \frac{w (\typ_{s_2} (A) ) w(\typ_{s_2} (B) ) w (\typ_{s_2} (C) ) w(\typ_{s_2} (D) )}{w(\typ_{s_1} (A)) w(\typ_{s_1} (B)) w(\typ_{s_1} (C)) w(\typ_{s_1} (D))}.
$$
This equality can be directly checked from the definitions of the jump rates. We give this verification in the form of a table below.
\begin{equation*}
\begin{pmatrix}
\frac{p(s_1 \to s_2)}{p(\bar s_2 \to \bar s_1)} & \mbox{Figure \ref{fig:rightJump-1}} & \mbox{Figure \ref{fig:rightJump-2}} & \mbox{Figure \ref{fig:rightJump-3}} & \mbox{Figure \ref{fig:rightJump-4}} \\
\mbox{Figure \ref{fig:rightJump-1}} & 1 & \frac{w(\rr) w(\uu)}{w(\ru) w(\ur)} & \frac{w(\ru) w(\ur)}{w(\all) w(\nul)} & \frac{w(\rr) w(\uu)}{w(\all) w(\nul)} \\
\mbox{Figure \ref{fig:rightJump-2}} & \frac{w(\ru) w(\ur)}{w(\rr) w(\uu)} & 1 & \frac{w(\ru)^2 w(\ur)^2}{w(\all) w(\nul) w(\rr) w(\uu)} & \frac{w(\ru) w(\ur)}{w(\all) w(\null)} \\
\mbox{Figure \ref{fig:rightJump-3}} & \frac{w(\all) w(\nul)}{w(\ru) w(\ur)} & \frac{w(\all) w(\nul) w(\rr) w(\uu)}{w(\ru)^2 w(\ur)^2} & 1 & \frac{w(\rr) w(\uu)}{w(\ru) w(\ur)} \\
\mbox{Figure \ref{fig:rightJump-4}} & \frac{w(\all) w(\nul)}{w(\rr) w(\uu)} & \frac{w(\all) w(\nul)}{w(\ru) w(\ur)} & \frac{w(\ru) w(\ur)}{w(\rr) w(\uu)} & 1 \\
\end{pmatrix}
\end{equation*}
The rows of this table correspond to different types of vertices in A and B. The columns correspond to different types of vertices in C and D after the flip. For the entries we simply calculate the quotient of jump rates (which are defined in the caption to the corresponding figures). On the other hand, it can be directly checked that the numerator of the obtained result corresponds to the weights of vertices which appear after the transform, while the denominator correspond to the vertices which disappear after the transform.

For example, if the vertices $A$ and $B$ are as in Figure \ref{fig:rightJump-3}, and the vertices $C$ and $D$ after the flip are as in Figure \ref{fig:rightJump-1}, then in $s_1$ (before the jump) the vertices $A$, $B$, $C$, and $D$ have types $\ru$, $\ur$, $\ur$, $\rr$, respectively, and in $s_2$ (after the jump) they have types $\rr$, $\all$, $\nul$, $\ur$, respectively. Therefore, the quotient of the product of all weights is equal to $\frac{w(\all) w(\nul)}{w(\ru) w(\ur)} $, which equals the corresponding quotient of jump rates.

The verification in other cases is analogous.

The case of jumps to the left can be proved in a similar way. However, it also follows from the case of jumps to the right. Indeed, in the case of jumps to the left all changes of types of vertices in $A$, $B$, $C$, $D$ are inverse to the changes that happen with the corresponding jump to the right (recall that the $r$-admissible pairs and $l$-admissible pairs are in the natural bijection). On the other hand, the jump rates for corresponding jumps are inverse to each other as well. Therefore, we have the desired equality in the case of jumps to the left as well.

\end{proof}

For two types of allowed vertices $X$, $Y$ and a state $s$, let us denote by $N_s (X,Y)$ the number of $h$-adjacent pairs such that the left vertex from this pair has type $X$ and the right vertex from this pair has type $Y$. We are interested in these numbers because the jump rates of our dynamics are determined by pairs of $h$-adjacent vertices.

\begin{lemma}
\label{lem:technical}
For any state $s \in \s$, we have
\begin{equation}
\label{eq:numb-vert-techn}
N_s (\ru, \nul) + N_s ( \rr, \ru ) = N_s (\nul, \ur) + N_s (\ur, \rr).
\end{equation}
\end{lemma}
\begin{proof}
All quantities in \eqref{eq:numb-vert-techn} depend on two h-adjacent vertices. Recall that our model lives on the grid $\{0, 1, \dots, M-1\} \times \{0,1, \dots, N-1 \}$. Let us consider the part of our grid $\mathrm{Th}_k := \{k, k+1\} \times \{0,1, \dots, N-1\}$, for some $0 \le k < M$. Let us prove that we have the desired equality of pairs of h-adjacent vertices inside $\mathrm{Th}_k$.

We say that $i \in \{0,1, \dots,N-1\}$ is a \textit{good} number if two vertical edges of the form $(k,i) \to (k,i+1)$ and $(k+1,i) \to (k+1,i+1)$ do not contain arrows from $s$. We say that a \textit{block} is a collection of good numbers $\{a,a+1, a+2, \dots, b\}$ such that $a-1$ and $b+1$ are not good numbers. See Figure \ref{fig:blocks} for an example.

\begin{figure}
\includegraphics[height=12cm]{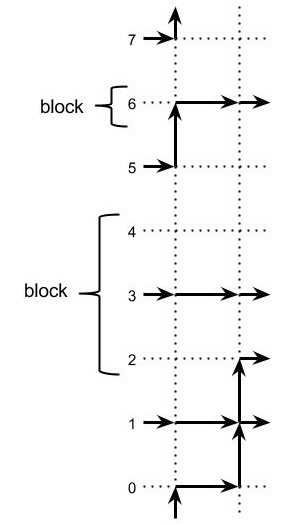}
\caption{An example of configuration in $\Th_k$ and blocks inside it.}
\label{fig:blocks}
\end{figure}

If all numbers from $0$ to $N-1$ are good, then we obviously have the desired equality. Otherwise, note that for each block $\{a,a+1, a+2, \dots, b\}$ the vertices $(k,b+1)$ and $(k+1,b+1)$ must be of the form $(\ru, \nul)$ or $(\rr, \ru)$. On the other hand, the vertices $(k,a)$ and $(k+1,a)$ must be of the form $(\nul, \ur)$ or $(\ur, \rr)$. Therefore, the number of blocks equals the number of configurations both in the left-hand side and in the right-hand side of \eqref{eq:numb-vert-techn}. Thus, we have this equality inside each $\mathrm{Th}_k$. Summing over all $k$, we arrive at the statement of the lemma.
\end{proof}

After the flip of a state, a pair of $h$-adjacent vertices transforms into another pair of $h$-adjacent vertices. The next lemma compares the numbers of a certain pairs of types of $h$-adjacent vertices in a state and in its flip.

\begin{lemma}
\label{lem:count-of-pairs}
Let
$$
A:= N_s (\ru, \nul) - N_s (\nul, \ur).
$$
Then we have
\begin{equation}
\label{eq:lemma4-eq1}
N_s (\uu, \ur) - N_s (\ru, \uu) = A, \qquad  N_s (\nul, \uu) - N_s (\uu, \nul) = A,
\end{equation}
and
\begin{equation}
\label{eq:lemma4-eq2}
N_s (\ur, \rr) - N_s (\rr, \ru) = A, \quad  N_s (\all, \ru) - N_s (\ur, \all) = A, \quad  N_s (\rr, \all) - N_s (\all, \rr) = A.
\end{equation}
\end{lemma}
\begin{proof}
In the proof of Lemma \ref{lem:flip-weight} above we showed that $N_s (\ru) = N_s (\ur)$. We have
\begin{align*}
N_s (\ru) = N_s (\ru, \nul) + N_s (\ru, \uu) + N_s (\ru, \ur), \\
N_s (\ur) = N_s (\nul, \ur) + N_s (\uu, \ur) + N_s (\ru, \ur).
\end{align*}
Subtracting the second equality from the fist one, we obtain $N_s (\uu, \ur) - N_s (\ru, \uu) = A$. Similarly,
\begin{align*}
N_s (\uu) = N_s (\nul, \uu) +N_s (\ru, \uu) +N_s (\uu, \uu), \\
N_s (\uu) = N_s (\uu, \nul) + N_s (\uu, \ur) + N_s (\uu, \uu).
\end{align*}
Subtracting the second equality from the first one we obtain the remaining equality from \eqref{eq:lemma4-eq1}. Applying already proved equalities \eqref{eq:lemma4-eq1} to the dual to $s$ state, we obtain that
$$
N_s (\ur, \rr) - N_s (\rr, \ru) = B, \quad  N_s (\all, \ru) - N_s (\ru, \all) = B, \quad  N_s (\rr, \all) - N_s (\all, \rr) = B,
$$
for some $B \in \mathbb Z$. Finally, Lemma \ref{lem:technical} gives $A=B$.
\end{proof}

Now we are ready to prove our main result.
\begin{theorem}
\label{th:main}
The Markov dynamics $\mathcal D$ preserves the Gibbs measure for the 6-vertex model on a torus.
\end{theorem}
\begin{proof}
Let $\mathcal P^+$ be the set of all pairs of states $s_1, s_2 \in \s$ such that $p(s_1 \to s_2)$ is strictly positive. It is enough to prove that for each $s \in \s$ we have
\begin{equation}
\label{eq:in-equals-out}
\sum_{s_1: (s, s_1) \in \mathcal P^+} w(s) p(s \to s_1) = \sum_{s_2: (s_2,s) \in \mathcal P^+} w(s_2) p(s_2 \to s).
\end{equation}
Using Lemma \ref{lem:flip-main}, we obtain
$$
\sum_{s_2: (s_2,s) \in \mathcal P^+} w(s_2) p(s_2 \to s) = \sum_{\bar s_2: (\bar s, \bar s_2) \in \mathcal P^+} w(\bar s) p( \bar s \to \bar s_2).
$$
Taking into account Lemma \ref{lem:flip-weight} and the equation above, we reduce \eqref{eq:in-equals-out} to
\begin{equation}
\label{eq:in-equals-in}
\sum_{s_1: (s, s_1) \in \mathcal P^+} p(s \to s_1) = \sum_{\bar s_2: (\bar s, \bar s_2) \in \mathcal P^+} p( \bar s \to \bar s_2).
\end{equation}
Recall that the jump rates are determined in an explicit and rather simple way, and
note that Lemma \ref{lem:count-of-pairs} provides us with the required information for comparing the left-hand and the right-hand sides of \eqref{eq:in-equals-in}. Indeed, all quantities in Lemma 4 compare the number of the same admissible pairs in a state and its flip. Taking into account all admissible pairs, we get
\begin{multline}
\label{eq:final-count}
\sum_{s_1: (s, s_1) \in \mathcal P^+} p(s \to s_1) -  \sum_{\bar s_2: (\bar s, \bar s_2) \in \mathcal P^+} p( \bar s \to \bar s_2)
\\ = \left( A \sqrt{ \frac{ w(\rr) w(\uu)}{w(\all) w(\nul)}} - A \frac{w (\ru) w(\ur)}{\sqrt{w(\all) w(\nul) w(\rr) w(\uu)}} + A \sqrt{ \frac{ w(\all) w(\nul)}{w (\rr) w(\uu)}} \right) \\ + \left( - A \sqrt{ \frac{ w(\all) w(\nul)}{w (\rr) w(\uu)}} + A \frac{w (\ru) w(\ur)}{\sqrt{w(\all) w(\nul) w(\rr) w(\uu)}} - A \sqrt{ \frac{ w(\rr) w(\uu)}{w(\all) w(\nul)}} \right).
\end{multline}
The first parentheses come from the jumps to the right; the second parentheses come from the jumps to the left.
Indeed, the first term in the first parentheses comes from the comparison of $h$-adjacent vertices $(\ru, \nul)$ (Figure \ref{fig:rightJump-1}) in a state and in its flip, the second term in the first parentheses comes from $(\uu, \nul)$ (Figure \ref{fig:rightJump-2}), and the third term from the first parentheses comes from $(\uu, \ru)$ (Figure \ref{fig:rightJump-4}). The pair $(\ru, \ur)$ (Figure \ref{fig:rightJump-3}) does not contribute because it is invariant under the flip.

Analogously, the first term in the second parentheses comes from the comparison of $h$-adjacent vertices $(\rr,\ru)$ (Figure \ref{fig:leftJump-1}) in a state and in its flip, the second term in the second parentheses comes from $(\rr, \all)$ (Figure \ref{fig:leftJump-3}), and the third term from the first parentheses comes from $(\ur, \all)$ (Figure \ref{fig:leftJump-4}). The pair $(\ur, \ru)$ (Figure \ref{fig:leftJump-2}) does not contribute because it is invariant under the flip.

It is clear that the total sum in \eqref{eq:final-count} equals 0, which concludes the proof of the theorem.
\end{proof}

\begin{remark}
The computations above shows that in the general case we need both jumps to the right and to the left in order to obtain the invariance of the Gibbs measure. However, at the free-fermion point (i.e., in the case $w(\ru) w(\ur) -w(\all) w(\nul) - w(\rr) w(\uu) = 0$) the dynamics can be decomposed into two parts; the part of the dynamics that corresponds to jumps only to the right still preserves the Gibbs measure (and the part of the dynamics that only jumps to the left as well). This is clearly visible from \eqref{eq:final-count} --- both parentheses vanish.
\end{remark}

\section{Appendix: The origin of dynamics}

In this section we briefly describe the origin of the dynamics $\mathcal D$.

Let $u$ and $q$ be two reals, and consider the following weights of the 6 types of vertices:
\begin{multline}
\label{eq:weights-Bor}
w(\nul)=1, \qquad w(\all) = \frac{u-\sqrt{q}}{1 - \frac{u}{\sqrt{q}} }, \qquad w( \rr) = \frac{u-\frac{1}{\sqrt{q}}}{1 - \frac{u}{\sqrt{q}} }, \\
w(\uu) = \frac{1 - \sqrt{q} u}{1 - \frac{u}{\sqrt{q}} }, \qquad w(\ru) = \frac{1-q}{1 - \frac{u}{\sqrt{q}} }, \qquad w(\ur) = \frac{(1-q^{-1})u}{1 - \frac{u}{\sqrt{q}} }.
\end{multline}
These weights coincide with weights from \cite[Definition 2.1]{Bor} with $s=q^{-1/2}$.

Let $\lambda = (\lambda(1), \lambda(2), \dots, \lambda(N))$, $\la (1) > \la(2) > \dots > \la (N)$, be a $N$-tuple of decreasing nonnegative integers.
Let $\Si_N$ be the set of all such $N$-tuples ($\la$'s are called \emph{signatures}).
We say that $\la_{N} \in \Si_{N}$ and $\la_{N+1} \in \Si_{N+1}$ \textit{interlace} and write $\la_N \prec \la_{N+1}$ if
$$
\la_{N+1} (1) \ge \la_N (1) \ge \la_{N+1} (2) \ge \dots \ge \la_N (N) \ge \la_{N+1} (N+1).
$$
Define a \textit{triangular array} of particles as the collection $\la_1 \prec \la_2 \prec \dots \prec \la_N$ (it is usually referred to as a \textit{Gelfand-Tsetlin scheme} due to its relation to the branching of the irreducible representations of the unitary groups).

For $\lambda \in \Si_N$ let $F_{\lambda}$ and $G_{\lambda}$ be the rational symmetric functions introduced in \cite{Bor}. They obey an analog of the Cauchy identity (\cite[Equation (4.6)]{Bor} and its skew versions). This allows us to apply a general construction of dynamics on triangular arrays which samples from a certain probability measure. The idea was introduced in \cite{BorFer}, \cite{Bor-schur} for the Schur symmetric polynomials and later applied in \cite[Section 2.3]{BorCor} to Macdonald polynomials. The exposition in \cite[Section 2.3]{BorCor} can be applied to the case of $F$- and $G$- functions directly.

One of the main properties of such a dynamics is that it samples from a probability measure on a triangular array. Let $K$ be a positive integer which will play a role of time in our dynamics, and let $x$ be a real fixed parameter. It is possible to sample a probability measure such that for each $N$ the projection to $\Si_N$ of this measure has the form
$$
const \cdot F_{\lambda} ( \underbrace{x,x,\dots, x}_{N}) G_{\lambda} (\underbrace{u,u,\dots, u}_K)
$$
(in fact, the joint probabilities for different $N$ also have a prescribed form which is similar to the ascending Macdonald process, see \cite[Section 1.1]{BorCor}). Moreover, this dynamics allows to make a step $K \to K+1$ (that is, starting with the probability measure corresponding to $K$ obtain the probability measure corresponding to $K+1$) by rules which involve only local interaction of particles of the triangular array. We denote by $\la_N^{K}$ the (random) signature from $\Si_N$ at time $K$.

Due to the combinatorial formula for $F$- functions (see \cite[Definition 3.1]{BorFer}) this dynamics has a rather simple form. Let us further assume that $u= q^{-1/2} + \eps$, for $0<\eps \ll 1$. Then the vertices $\uu$ and $\rr$ have small weights proportional to $\eps$ (see \eqref{eq:weights-Bor}). This implies that with large probability the step $K \to K+1$ of our dynamics on each level looks like $\la_N^{K+1} (i) = \la_N^{K} (i) +1$, $i=1,2, \dots, N$, and the events $\la_N^{K+1} (i) = \la_N^{K} (i)$ and $\la_N^{K+1} (i) = \la_N^{K} (i)+2$ happen with probability proportional to $\eps$, while all other events can be considered as impossible (these two $\eps$-probability events eventually give rise to the left and right jumps in dynamics $\mathcal D$). We are interested in the limit of this dynamics for $N=[\eps^{-1}]$, $\eps \to 0$. If we consider this dynamics ``in the bulk'', then our formulas for probability measures (heuristically) imply the preservation of the Gibbs measure on the set of paths. The dynamics on the triangular array is governed by local interactions only. Therefore, it is reasonable to assume that ``boundary conditions'' do not affect the dynamics in the bulk, and that after the limit $\eps \to 0$ the dynamics in the bulk will preserve the Gibbs measure on a six-vertex model on torus. This paper describes the limit dynamics and proves that it has such a property.


\begin{thebibliography}{50}

\bibitem{Bor} A. Borodin. \textit{On a family of symmetric rational functions}, preprint, {\tt arXiv:1410.0976}.

\bibitem{Bor-schur} A. Borodin. \textit{Schur dynamics of the Schur processes}, Adv. Math. \textbf{228} (2011), 2268--2291. 

\bibitem{BorodinBufetovOlshanski}
A.~Borodin, Al. Bufetov, and G.~Olshanski.
\newblock \emph{Limit shapes for growing extreme characters of $U(\infty)$},
\newblock {Ann. Appl. Probab.}
\textbf{25} (2015), 2339--2381.

\bibitem{BorCor} A.~Borodin and I.~Corwin. \textit{Macdonald processes}, Prob. Theory Rel. Fields
  \textbf{158} (2014), 225--400. 

\bibitem{BorFer} A.~Borodin and P.~L.~Ferrari. \textit{Anisotropic growth of random surfaces in 2+1 dimensions}, Comm. Math. Phys. \textbf{325} (2014), 603--684. 

\bibitem{BorFer2}
A.~Borodin and P.~L.~Ferrari, \textit{Anisotropic KPZ growth in 2+1 dimensions: 
Fluctuations and covariance
structure}, J. Stat. Mech. (2009) P02009.

\bibitem{BorPetLectures}
A.~Borodin and L.~Petrov.
\newblock \emph{Integrable probability:
From representation theory to Macdonald processes},
\newblock{{Prob. Surveys} \textbf{11} (2014), 1--58}.

\bibitem{Corwin-KPZ}
I.~Corwin.
\newblock \emph{The Kardar-Parisi-Zhang equation and universality class},
\newblock {Random Matrices Theory Appl.} \textbf{1} (2012).
\newblock arXiv:1106.1596 [math.PR].

\bibitem{CorwinToninelli} I.~Corwin and F.~L.~Toninelli. \textit{Stationary measure of the driven two-dimensional $q$-Whittaker particle system on the torus}, Preprint, 2015, arXiv:1509.01605. 

\bibitem{FerrariSpohn} P.~L.~Ferrari and H.~Spohn. \textit{Random growth models}, The Oxford Handbook of Random Matrix
Theory, G. Akemann, J. Baik and P. Di Francesco (eds.), 2011.

\bibitem{GatesWestcott} D.~J.~Gates and M.~Westcott. \textit{Stationary states of crystal growth in three dimensions}, J. Stat. Phys. \textbf{81} (1995), 681--715.

\bibitem{Kenyon-Okounkov-Sheffield} R.~Kenyon, A.~Okounkov, and S.~Sheffield. \textit{Dimers and amoebae}, Ann. Math. \textbf{163}, 1019--1056 (2006).

\bibitem{KipnisLandim} C.~Kipnis and C.~Landim. \textit{Scaling Limits of Interacting Particle Systems}. Grundlehren der mathematischen Wissenschaften, Vol. 320, Springer, 1999.

\bibitem{PalamarchukReshetikhin}
K.~Palamarchuk and N.~Reshetikhin.
\newblock  \emph{The 6-vertex model with fixed boundary conditions.}
\newblock In: Workshop Bethe Ansatz: 75 Years Laters, Brussels, Belgium, 2006.
\newblock  	arXiv:1010.5011 [math-ph].

\bibitem{Sheffield} S.~Sheffield. \textit{Random surfaces}, Ast\'{e}risque (2005), no. 304.

\bibitem{Toninelli}
F.~L.~Toninelli. \textit{A (2+1)-dimensional growth process with explicit stationary measures},
Preprint, 2015, arXiv:1503.05339.


\end{thebibliography}
\end{document}